\documentclass[11pt,onecolumn,draftcls]{IEEEtran}
\usepackage{amsmath,epsfig}
\usepackage{amssymb}
\usepackage{amsmath}
\usepackage{cases} 
\usepackage{amssymb,amsmath,cite}
\usepackage{epsfig}
\usepackage{color}
\usepackage{slashbox}
\usepackage{bm}

\renewcommand{\frac}{\dfrac}

\newcommand{\K}{{\cal K}}
\newcommand{\M}{{\cal M}}
\newcommand{\N}{{\cal N}}
\newcommand{\R}{{\cal R}}

\newcommand{\X}{{\cal X}}
\newcommand{\Y}{{\cal Y}}
\newcommand{\Z}{{\cal Z}}
\renewcommand{\S}{\cal S}
\newcommand{\SI}{\mbox{SINR}}

\newtheorem{dingli}{Theorem~}[section]
\newtheorem{yinli}{Lemma~}[section]

\newenvironment{proof}[1][Proof]{\begin{trivlist}
\item[\hskip \labelsep {\bfseries #1}]}{\end{trivlist}}

\begin{document}
  \title{On the Complexity of Joint Subcarrier and Power Allocation for Multi-User OFDMA Systems
  \author{Ya-Feng Liu and Yu-Hong Dai}
\thanks{This work was supported in part by the China National Funds for Distinguished Young Scientists, Grant 11125107, and the National Natural Science Foundation, Grants 11331012, 11301516, and 81173633, and in part by the CAS Program for Cross \& Cooperative Team of the Science \& Technology Innovation.}
\thanks{Copyright (c) 2013 IEEE. Personal use of this material is permitted. However, permission to use this material for any other purposes must be obtained from the IEEE by sending a request to pubs-permissions@ieee.org.}
 \thanks{Y.-F.~Liu and Y.-H. Dai are with the State Key Laboratory
of Scientific and Engineering Computing, Institute of Computational
Mathematics and Scientific/Engineering Computing, Academy of
Mathematics and Systems Science, Chinese Academy of Sciences,
Beijing, 100190, China (e-mail:
{\{{yafliu,\,dyh}\}@lsec.cc.ac.cn}).}
 }
 \maketitle
 \begin{abstract}
     \boldmath Consider a multi-user Orthogonal Frequency Division Multiple Access (OFDMA) system where multiple users share multiple discrete subcarriers, but at most one user is allowed to transmit power on each subcarrier. To adapt fast traffic and channel fluctuations and improve the spectrum efficiency, the system should have the ability to dynamically allocate subcarriers and power resources to users. Assuming perfect channel knowledge,
     ~two formulations for the joint subcarrier and power allocation problem are considered in this paper: the first is to minimize the total transmission power subject to quality of service constraints and the OFDMA constraint, and the second is to maximize some system utility function (including the sum-rate utility, the proportional fairness utility, the harmonic mean utility, and the min-rate utility) subject to the total transmission power constraint per user and the OFDMA constraint. In spite of the existence of various heuristics approaches, little is known about the computational complexity status of the above problem. 
     This paper aims at filling this {theoretical} gap, i.e., characterizing the complexity of the joint subcarrier and power allocation problem for the multi-user OFDMA system. It is shown {in this paper} that both formulations of the joint subcarrier and power allocation problem are strongly NP-hard. The proof is based on a polynomial time transformation from the so-called $3$-dimensional matching problem. Several subclasses of the problem which can be solved to global optimality or $\epsilon$-global optimality in polynomial time are also identified. {These complexity results suggest that there are not polynomial time algorithms which are able to solve the general joint subcarrier and power allocation problem to global optimality (unless P$=$NP), and determining an approximately optimal subcarrier and power allocation strategy is more realistic in practice.}
     \end{abstract}

  \section{Introduction}
With the rapid growth of wireless data traffic, multi-user interference has
become a major bottleneck that limits the performance of the existing wireless communication
services. Mathematically, we can model such an interference-limited communication system as
a multi-user interference channel in which a number of linearly interfering transmitters
simultaneously send private data to their respective receivers. Exploiting time/space/frequency/code diversity are effective approaches to mitigate/manage multi-user interference. For instance, when the transmitters and/or receivers are equipped with multiple antennas, a joint optimization of the physical layer
transmit-receive beamforming vectors for all users can efficiently mitigate multi-user interference\cite{Per-antenna,coordinated,Coordition,ICC}; when all transmitters and receivers are equipped with a single antenna, one way to control/mitigate multi-user interference is to impose certain frequency restrictions or transmission power limits on the frequency resources used by each transmitter\cite{dynamic,complexity}. In particular, {Orthogonal Frequency Division Multiple Access (OFDMA) is a form of multi-carrier transmission and is suited for frequency selective
channels and high data rates. This technique effectively decomposes a frequency-selective wide-band channel into a group of
non-selective narrowband subchannels (subcarriers), which makes it robust against large delay spreads by preserving orthogonality
in the frequency domain. Moreover, the ingenious introduction of cyclic redundancy at the transmitter reduces the
complexity to only FFT processing and one tap scalar equalization at the receiver\cite{goldsmith}.}

 Conventional OFDMA schemes \emph{preassign} subcarriers to users in a nonoverlapping way, thus users (transmitting on different subcarriers) cause no interference to each other. Although the OFDMA scheme is well suited to be used in a high-speed communication context where quality of service is a major concern, it can lead to inefficient bandwith utilization. This is because that the preassignment of subcarries can not adapt traffic load and channel fluctuations in space and time; i.e., a subcarrier preassigned to a user can not be released to other users even if it is unusable when the user's channel conditions are poor.
To adapt these fluctuations and improve the overall system's throughput, OFDMA based subcarrier allocation networks such as Worldwide Interoperability for Microwave Access (WiMAX)\cite{wimax} and Long Term Evolution (LTE)\cite{ofdma} should be equipped with \emph{dynamic} subcarrier and power allocation algorithms. In particular, a dynamic OFDMA based subcarrier and power allocation algorithm is well suited for the \emph{dense} femtocell downlink system\cite{self}, where a large number of femtocells close to each other are deployed in a macrocell. 

The joint optimization of subcarrier and power allocations for the multi-user OFDMA system is a nonconvex problem, therefore various heuristics approaches have been proposed for this problem {\cite{dynamic,wong,partition,self,complexity,hayashi,computation,broadcast,JSAC,Adaptive,outage,sub-optimal,personal,cognitive,newa,newb,newc,newd,newe,newf,newg,newh,newi,newj,newk}}. 
Very recently, the authors in \cite{dynamic} proposed a dynamic joint frequency and
transmission power allocation scheme called enhanced Dynamic
Frequency Planning (eDFP). eDFP is a hybrid centralized and distributed architecture,
where a central broker first dynamically partitions subcarriers among
neighboring cells so that the long-term cell-edge inter-cell
interference is minimized, and then each cell independently
allocates subcarriers and transmission power to its users in a way that its
total transmission power is minimized.

Notice that nonconvex
optimization problems are difficult to solve in general. However, not all nonconvex problems
are hard to handle since the lack of convexity may be due to an inappropriate
formulation, and many nonconvex optimization problems indeed admit a convex reformulation \cite{complexity,precoding,simo}. Therefore, nonconvexity of the joint subcarrier and power allocation problem for the multi-user OFDMA system does not imply that it is computationally intractable (strongly NP-hard in terms of computational complexity theory\cite{Complexitybook,combook2,combook3,combinatorial}). The aim of this paper is to characterize the computational complexity of the joint subcarrier and power allocation problem for the multi-user OFDMA system; i.e., to categorize when the problem is (strongly) NP-hard and when it is polynomial time solvable.

In fact, the dynamic spectrum management problem \emph{without the OFDMA constraint} ({at most one user is allowed to transmit power on each subcarrier}) has been extensively studied in \cite{complexity,hayashi,gap,dual,dual2}. It is shown in \cite{complexity,hayashi} that the dynamic spectrum management problem is (strongly) NP-hard when the number of subcarriers is greater than two, or when the number of users is greater than one. However, the analysis of these results is highly dependent on the crosstalk channel gains among users; i.e., some of the crosstalk channel gains are assumed to be large enough, and some of them are assumed to be zero. We shall see late (in Section \ref{sec-model}) that the magnitude of crosstalk channel gains has no influence on the user's transmission rate if all users are required to transmit power in an orthogonal manner. This makes the two problems (the dynamic spectrum management problem without the OFDMA constraint in \cite{complexity} and the joint subcarrier and power allocation problem with the OFDMA constraint considered in this paper) sharply different from each other.
 An interesting result shown in \cite{hayashi} is that the optimal solution of the two-user \emph{sum-rate} maximization problem is automatically OFDMA if the crosstalk channel gains of the two users on each subcarrier are large enough. In addition, based on the fact that the duality gap of the dynamic spectrum management problem is zero as the number of subcarriers goes to infinity (regardless of the convexity of the objective function)\cite{dual,gap,complexity}, dual decomposition algorithms have been proposed in \cite{dual,dual2} for the nonconvex optimization problem of multi-carrier systems, again without considering the OFDMA constraint.


In this paper, we focus on the characterization of the computational complexity status of the joint subcarrier and power allocation problem for the \emph{multi-user OFDMA system}. In particular, we consider two formulations of the joint subcarrier and power allocation problem. The first one is the problem of minimizing the total transmission power in the system subject to all users' quality of service constraints, all users' power budget constraints per subcarrier, and the OFDMA constraint. The second one is the problem of maximizing the system utility (including the sum-rate utility, the proportional fairness utility, the harmonic mean utility, and the min-rate utility) while the total transmission power constraint of each user, individual power constraints on each subcarrier, and the OFDMA constraint are respected. {The main contributions of this paper are twofold. First, we show that the aforementioned two formulations of the joint subcarrier and power allocation problem are strongly NP-hard. The proof is based on a polynomial time transformation from the $3$-dimensional matching problem. The strong NP-hardness results suggest that for a given OFDMA system, computing the optimal subcarrier and power allocation strategy is generally intractable. Thus, instead of insisting on finding an efficient algorithm that can find the global optimum of the joint subcarrier and power allocation problem, one has to settle with less ambitious goals, such as finding high quality approximate solutions or locally optimal solutions of the problem in polynomial time. Second, we also identify several subclasses of the joint allocation problem which can be solved to global optimality or $\epsilon$-global optimality in polynomial time. We therefore clearly delineate the set of computationally tractable problems within the general class of NP-hard joint subcarrier and power allocation problems. Specifically, we show in this paper that, when there is only a single user in the system or when the number of subcarriers and the number of users are equal to each other, the total transmission power minimization problem is polynomial time solvable; when there is only a single user, the aforementioned four utility maximization problems are all polynomial time solvable.}

{The rest of this paper is organized as follows. In Section \ref{sec-model}, we introduce the system model and give the two formulations of the joint subcarrier and power allocation problem for the multi-user OFDMA system. In Section \ref{sec-hard}, we first give a brief introduction to computational complexity theory\cite{Complexitybook,combook2,combook3,combinatorial} and then address the computational complexity of the joint subcarrier and power allocation problem. In particular, we show that the aforementioned two formulations of the joint subcarrier and power allocation problem are generally strongly NP-hard. 
Several subclasses of the joint allocation problem which are polynomial time solvable are identified in Section \ref{sec-easy}. Finally, the conclusion is drawn in Section \ref{sec-conclusion}.


}

%

\section{System Model and Problem Formulation}\label{sec-model}
{In this section, we introduce the system model and problem formulation.} Consider a multi-user OFDMA system, where there are $K$ users (transmitter-receiver pairs) sharing $N$ subcarriers. Throughout the paper, we assume that $N\geq K$; i.e., the number of subcarriers is greater than or equal to the number of users. Otherwise, the OFDMA constraint is infeasible.

Denote the set of users and the set of subcarriers by $\K=\left\{1,2,...,K\right\}$ and $\N=\left\{1,2,...,N\right\}$, respectively. For any $k\in\K$ and $n\in\N$, suppose $s_k^n\in\mathbb{C}$ to be the symbol that transmitter $k$ wishes to send to receiver $k$ on subcarrier $n$, then the received signal $\hat s_k^n$ at receiver $k$ on subcarrier $n$ can be expressed by
$$\hat s_k^n=h_{k,k} s_k^n+\sum_{j\neq k} h_{k,j}^n s_j^n + z_k^n,$$ where $h_{k,j}^n\in\mathbb{C}$ is the channel coefficient from transmitter $j$ to receiver $k$ on subcarrier $n$ and $z_k^n\in\mathbb{C}$ is the additive white Gaussian
noise (AWGN) with distribution $\cal{CN}$$({0}, \eta_k^n).$
%
Denoting the power of $s_k^n$ by $p_k^n$; i.e., $p_k^n:=|s_k^n|^2$, the received power at receiver $k$ on subcarrier $n$ is given by
$$\alpha_{k,k}^np_k^n+\sum_{j\neq k}\alpha_{k,j}^np_j^n+\eta_k^n,~k\in\K,~n\in\N,$$ where $\alpha_{k,j}^n:=|h_{k,j}^n|^2$ 
stands for the channel gain from transmitter $j$ to receiver $k$ on subcarrier $n.$ Treating interference as noise, we can write the SINR of receiver $k$ on subcarrier $n$ as
$$\SI_k^n=\frac{\alpha_{k,k}^np_k^{n}}{\eta_k^n+\sum_{j\neq k} \alpha_{k,j}^np_j^{n}},~k\in\K,~n\in\N,$$ and transmitter $k$'s achievable data rate $R_k$ {(bits/sec)} as \begin{equation}\label{rk}R_{k}=\sum_{n\in\N}{\log_2}\left(1+\SI_k^n\right),~k\in\K.\end{equation}


  {In this paper, we consider the joint subcarrier and power allocation problem for the multi-user OFDMA system. 
  Mathematically, a power allocation vector $\left\{p_k^n\right\}$ is said to satisfy the OFDMA property if the following equations hold true:
  $$p_k^np_j^n=0,~\forall~j\neq k,{~k,\,j\in\K,}~n\in\N.$$ The above equations basically say that at most one user is allowed to transmit power on each subcarrier. Therefore, the joint subcarrier and power allocation problem for the multi-user OFDMA system can be formulated as}
 \begin{equation}\label{problem}
 \begin{array}{cl}
\displaystyle \min_{\left\{p_k^n\right\}} & \displaystyle \sum_{k\in\K}\sum_{n\in\N}p_k^{n}\\[15pt]
\text{s.t.} & R_{k}\geq \gamma_{k},~k\in\K, \\[3pt]
    & P_k^n\geq p_k^n\geq 0,~k\in\K,~n\in\N,\\[3pt]
    & p_k^np_j^n=0,~\forall~j\neq k,{~k,\,j\in\K,}~n\in\N,
    \end{array}
\end{equation}
where 
the objective function is the total transmission power of all users on all subcarriers,
$\gamma_k>0$ ($k\in\K$) is the desired transmission rate target of user $k,$ $P_k^n$ is the transmission power
budget of user $k$ on subcarrier $n$, and the last constraint is the OFDMA constraint.

Due to the OFDMA constraint, we know for any $k\in\K$ and $n\in\N$
 $$\SI_k^n=\frac{\alpha_{k,k}^np_k^{n}}{\eta_k^n+\sum_{j\neq k} \alpha_{k,j}^np_j^{n}}=\frac{\alpha_{k,k}^np_k^{n}}{\eta_k^n}.$$
In fact, if $p_k^n=0$, the above equality holds trivially; while if $p_k^n> 0,$ it follows from the OFDMA constraint 
that $p_j^{n}=0~(\forall~j\neq k)$ and thus $\sum_{j\neq k} \alpha_{k,j}^np_j^{n}=0,$ which shows that the above equality holds
as well. Thus, problem \eqref{problem} is equivalent to
\begin{equation}\label{problem-simp}
\begin{array}{cl}
\displaystyle \min_{\{p_k^n\}} & \displaystyle \sum_{k\in\K}\sum_{n\in\N}p_k^{n}\\[10pt]
\text{s.t.} & \displaystyle\sum_{n\in\N}\log_2\left(1+\frac{\alpha_{k,k}^np_k^{n}}{\eta_k^n}\right)\geq \gamma_{k},~k\in\K, \\[12pt]
    & P_k^n\geq p_k^n\geq 0,~k\in\K,~n\in\N,\\[3pt]
    & p_k^np_j^n=0,~\forall~j\neq k,{~k,\,j\in\K,}~n\in\N.
    \end{array}
\end{equation}

By introducing a group of binary variables $x_k^n~(k\in\K,n\in\N),$ problem \eqref{problem-simp} can be reformulated as
 \begin{equation}\label{binary}
 \begin{array}{cl}
\displaystyle \min_{\{p_k^n\}, \{x_k^n\}} & \displaystyle \sum_{k\in\K}\sum_{n\in\N}p_k^{n}\\[10pt]
\text{s.t.} & \displaystyle \sum_{n\in\N}\log_2\left(1+\frac{\alpha_{k,k}^np_k^{n}}{\eta_k^n}\right)\geq \gamma_{k},~k\in\K, \\[12pt]
    & x_k^nP_k^n\geq p_k^n\geq 0,~k\in\K,~n\in\N,\\[3pt]
    & x_k^n\in\left\{0,1\right\},~k\in\K,~n\in\N,\\[3pt]
    & \displaystyle\sum_{k\in\K} x_k^n\leq 1,~n\in\N,
    \end{array}
\end{equation} where the binary variable $x_k^n=1$ if user $k$ transmits power on subcarrier $n,$ or $x_k^n=0$ otherwise. The last constraint $\sum_{k\in\K} x_k^n\leq 1~(n\in\N)$ stands for the OFDMA constraint.

Problem \eqref{binary} can be dealt with by a two-stage approach. At the first stage, we solve the subcarrier allocation problem; i.e., determining the binary variables $\left\{x_k^n\right\}$, which is equivalent to partitioning the set of subcarriers $\N=\left\{1,2,...,N\right\}$ into $K$ nonoverlapping groups $\left\{{\cal N}_k\right\}_{k=1}^K$. At the second stage, we solve the power allocation problem; i.e., solving $K$ decoupled power allocation problems \begin{equation}\label{power-allocation}
 \begin{array}{cl}
\displaystyle \min_{\left\{p_k^n\right\}_{n\in{\cal N}_k}} & \displaystyle \sum_{n\in {\cal N}_k}p_k^{n}\\[10pt]
\text{s.t.} & \displaystyle \sum_{n\in\N_k}\log_2\left(1+\frac{\alpha_{k,k}^np_k^{n}}{\eta_k^n}\right)\geq \gamma_{k}, \\[14pt]
    & P_k^n\geq p_k^n\geq 0,~n\in {\cal N}_k.\\
    \end{array} \end{equation}
    Problem \eqref{power-allocation} at the second stage is convex, and thus is easy to solve.

 To sum up, the joint subcarrier and power allocation problem can be equivalently formulated as \eqref{problem}, \eqref{problem-simp}, or \eqref{binary}.  Formulation \eqref{binary} is intuitive and is easy to understand, whereas formulation \eqref{problem-simp} is compact and is easy to analyze. The analysis of this paper is mainly based on \eqref{problem-simp}.

Besides the total transmission power minimization problem, we also consider the utility maximization problem for the multi-user OFDMA system, which can be expressed by
\begin{equation}\label{utility}
 \begin{array}{cl}
\displaystyle \max_{\{p_k^n\}} & \displaystyle H(R_1,R_2,...,R_K) \\[8pt]
\text{s.t.} 
    & \displaystyle \sum_{n\in\N}p_k^{n}\leq P_{k},~k\in\K, \\[12pt]
    & P_k^n\geq p_k^n\geq 0,~k\in\K,~n\in\N,\\[3pt]
    & p_k^np_j^n=0,~\forall~j\neq k,{~k,\,j\in\K,}~n\in\N,
    \end{array}
\end{equation}where $H(R_1,R_2,...,R_K)$ denotes the system utility function and $P_k~(k\in\K)$ is the power budget of transmitter $k.$
Four popular system utility functions are
\begin{itemize}
\item [-] {Sum-rate utility:}
$\displaystyle H_1(R_1,R_2,...,R_K)=\frac{1}{K}{\sum_{k=1}^K{R_k}},$
\item [-] {Proportional fairness utility:} $\displaystyle
H_2(R_1,R_2,...,R_K)=\left(\prod_{k=1}^KR_k\right)^{1/K},$\\[3pt]
\item [-] {Harmonic mean utility:}
$\displaystyle H_3(R_1,R_2,...,R_K)=K/\left(\sum_{k=1}^K R_k^{-1}\right),$\\[-3pt]
\item [-] {Min-rate utility:} $\displaystyle H_4(R_1,R_2,...,R_K)=\min_{1\leq k\leq
K}\left\{R_k\right\}.$
\end{itemize}
It is simple to see that \begin{align*}&H_1(R_1,R_2,...,R_K)\geq H_2(R_1,R_2,...,R_K)\\
\geq~& H_3(R_1,R_2,...,R_K)\geq H_4(R_1,R_2,...,R_K),\end{align*} and the equality holds true if and only if $R_1=R_2=\cdots=R_K.$

\section{Hard Cases} \label{sec-hard}

In this section, we show that both the total power minimization problem \eqref{problem-simp} and the system utility maximization problem \eqref{utility} are intrinsically intractable (strongly NP-hard in the sense of computational complexity theory \cite{Complexitybook,combook2,combook3,combinatorial}), provided that the ratio of the number of subcarriers and the number of users, that is $N/K$, is equal to any given constant number $c>1$.
To begin with, we briefly introduce complexity theory in Subsection \ref{sub-i}. Then we show problems \eqref{problem-simp} and \eqref{utility} are strongly NP-hard in Subsections \ref{sub-ii} and \ref{sub-iii}, respectively.

\subsection{A Brief Introduction to Complexity Theory}\label{sub-i}
In computational complexity theory{\cite{Complexitybook,combook2,combook3,combinatorial}}, a problem is said to be NP-hard if it is at least as hard as any problem in the class NP (problems that are solvable in Nondeterministic Polynomial time). The NP class includes well known problems like the $3$-colorability problem (which is to check whether the nodes of a given graph can be colored in three colors so that each pair of adjacent nodes are
colored differently). NP-complete problems are the hardest problems in NP in the sense that if any NP-complete problem is solvable in polynomial time, then each problem in NP is solvable in polynomial time. The $3$-colorability problem is NP-complete.

A problem is strongly NP-hard (strongly NP-complete) if it is NP-hard (NP-complete) and it can not be solved by a pseudo-polynomial time algorithm. {An algorithm that solves a problem is called a \emph{pseudo-polynomial} time algorithm if its time complexity function is bounded above by a polynomial function related to both of the {length} and the numerical values of the given data of the problem. This is in contrast to the polynomial time algorithm whose time complexity function depends only on the length of the given data of the problem.} The $3$-colorability problem is strongly NP-complete. However, not all NP-hard (NP-complete) problems are strongly NP-hard (strongly NP-complete). For instance, the partition problem is NP-hard but not strongly NP-hard.

Strongly NP-hard or NP-hard problems may not be in the class NP, but they are at least as hard as any NP-complete problem. It is widely believed that there can not exist a polynomial time algorithm to solve any NP-complete, NP-hard, or strongly NP-hard problem (unless P$=$NP). Thus, once an optimization problem is shown to be NP-hard, we can no longer insist on having an efficient algorithm that can find its global optimum in polynomial time. Instead, we have to settle with less ambitious goals, such as finding high quality approximate solutions or locally optimal solutions of the problem in polynomial time.

The standard way to prove an optimization problem is NP-hard is to establish the NP-hardness of its corresponding feasibility problem or decision problem. The latter is the problem to decide if the global minimum of the optimization problem is below a given threshold or not. The output of a decision problem is either true or false. The feasibility or decision version of an optimization problem is usually in the class NP. Clearly, the feasibility or decision version of an optimization problem is always easier than the optimization problem itself, since the latter further requires finding the global minimum value and the minimizer. Thus, if we show the feasibility or decision version of an optimization problem is NP-hard, then the original optimization problem must also be NP-hard. 

In complexity theory, to show a decision problem $\mathcal{B}$ is NP-hard,
we usually follow three steps: 1) choose a suitable known
NP-complete decision problem $\mathcal{A};$ 2) construct a \emph{polynomial
time} transformation from any instance of $\mathcal{A}$ to an instance of $\mathcal{B};$
3) prove under this transformation that any instance of problem
  $\mathcal{A}$ is true if and only if the constructed instance of problem $\mathcal{B}$ is true. Furthermore, if the chosen NP-complete problem $\mathcal{A}$ is strongly NP-complete, then problem $\mathcal{B}$ is strongly NP-hard.

{In the following two subsections, we show that problems \eqref{problem-simp} and \eqref{utility} are strongly NP-hard.}
\subsection{Strong NP-Hardness of Power Minimization Problem \eqref{problem-simp} when $N/K=c>1$}\label{sub-ii}

To analyze the computational complexity of problem \eqref{problem-simp}, we consider its feasibility problem. If the feasibility problem
is strongly NP-hard, so is the original optimization problem.

{\bf Feasibility Problem of \eqref{problem-simp}}. \ Given a set of transmission rate levels $\gamma_k,$ individual power budgets per subcarrier $P_k^n$, direct-link channel gains $\alpha_{k,k}^n,$ and noise powers $\eta_k^n,$ check whether there exists a subcarrier and power allocation strategy such that
\begin{equation}\label{problem2}\left\{\!\!\!\!
\begin{array}{cl}
    & \displaystyle \sum_{n\in\N}\log_2\left(1+\frac{\alpha_{k,k}^np_k^{n}}{\eta_k^n}\right)\geq \gamma_{k},~k\in\K, \\[12pt]
    & P_k^n\geq p_k^n\geq 0,~k\in\K,~n\in\N,\\[5pt]
    & p_k^np_j^n=0,~\forall~j\neq k,{~k,\,j\in\K,}~n\in\N.
    \end{array}\right.
\end{equation}

To analyze the computational complexity of the feasibility problem \eqref{problem2}, we choose the following strongly NP-complete problem (see \cite{Complexitybook}).

{\bf $3$-Dimensional Matching Problem with Size $K$.} \  {Let $\X, \Y,$ and $\Z$ be three different sets with $|\X|=|\Y|=|\Z|=K,$ and $\R$ be a subset of $\X\times\Y\times\Z$. The $3$-dimensional matching problem is to check whether there exists a match $\M\subset \R$ such that the following two conditions are satisfied:
\begin{itemize}
  \item [a)] for any two different tripes $(x_1,y_1,z_1)\in\M$ and $(x_2,y_2,z_2)\in\M,$ we have $x_1\neq x_2,$ $y_1\neq y_2,$ and $x_3\neq y_3;$
  \item [b)] $\X=\left\{x\,|\,(x,y,z)\in\M\right\}$,~$\Y=\left\{y\,|\,(x,y,z)\in\M\right\},$~and $\Z=\left\{z\,|\,(x,y,z)\in\M\right\}.$
\end{itemize}
%
}

Next, we first show that any instance of the $3$-dimensional matching problem corresponds to an instance of the transmission rate feasibility problem \eqref{problem2} when $N/K=2$.


\begin{yinli} [Basic Lemma] \label{thm-feasi}Checking the feasibility problem of \eqref{problem-simp} is strongly NP-hard when $N/K=2$. Thus, problem \eqref{problem-simp} itself is also strongly NP-hard when $N/K=2$. \end{yinli}
\begin{proof}
Assume that $N/K=2$. Consider any instance of the $3$-dimensional matching problem with $$\X=\left\{1_x,2_x,...,K_x\right\},$$$$ \Y=\left\{1_y,2_y,...,K_y\right\}, $$$$ \Z=\left\{1_z,2_z,...,K_z\right\},$$ and a relationship set $${\R=\left\{\left(k_x,j_y,l_z\right)\,|~k_x\in \X, j_y\in \Y, l_z\in \Z\right\}\subseteq\X\times\Y\times\Z},$$
  we construct a multi-user multi-carrier system where there are $K$ users (which correspond to set $\X$) and $2K$ subcarriers (which correspond to set $\Y \bigcup \Z$). More exactly, $\K=\X$ and $\N=\Y \bigcup \Z.$  Define
  \begin{equation}\label{s1s2}
  \begin{array}{cl}
  & {\S}_1=\left\{(k_x,j_y)\,|\,\left(k_x,j_y,l_z\right)\in \R\right\},\\[3pt]
  & {\S}_2=\left\{(k_x,l_z)\,|\,\left(k_x,j_y,l_z\right)\in \R\right\}.\\
  \end{array}\end{equation}
%
For each $k\in \K$, the power budgets per subcarrier $P_k^n$, the noise powers $\eta_k^n,$ and the direct-link channel gains $\alpha_{k,k}^n$ are
given by
    \begin{equation}\label{p}
  P_k^n=\left\{\begin{array}{lr}
       3, \quad \text{if}~n\in {\Y};\\[3pt]
       2, \quad \text{if}~n\in {\Z},\\[3pt]
    \end{array}\right.
  \end{equation}

  \begin{equation}\label{eta}
  \eta_k^n=\left\{\begin{array}{lr}
       1, \quad \text{if}~(k,n)\in {\S}_1;\\[3pt]
       2, \quad \text{if}~(k,n)\in {\S}_2;\\[3pt]
       3, \quad \text{if}~(k,n)\notin {\S}_1\bigcup {\S}_2,
    \end{array}\right.
  \end{equation}
  and \begin{equation}\label{alpha}
  \alpha_{k,k}^{n_1}=\alpha_{k,k}^{n_2}=\left\{\begin{array}{lr}
       1, \quad \text{if}~(k,n_1,n_2)\in {\R};\\[5pt]
       0.25, \quad \text{if}~(k,n_1,n_2)\notin {\R},
    \end{array}\right.
  \end{equation}
respectively.
Letting \begin{equation*}\label{target}\gamma_k=3,~k\in\K,\end{equation*} the corresponding instance of problem \eqref{problem2} is
\begin{equation}\label{construct}\left\{\!\!\!\!
\begin{array}{cl}
    & \displaystyle\sum_{n\in\N}\log_2\left(1+\frac{\alpha_{k,k}^np_k^{n}}{\eta_k^n}\right)\geq 3,~k\in\K, \\[12pt]
    & 3\geq p_k^n\geq 0,~k\in\K,~n\in\Y,\\[5pt]
    & 2\geq p_k^n\geq 0,~k\in\K,~n\in\Z,\\[5pt]
    & p_k^np_j^n=0,~\forall~j\neq k,{~k,\,j\in\K,}~n\in\N,
    \end{array}\right.
\end{equation} where $\eta_k^n$ and $\alpha_{k,k}^n$ are given in \eqref{eta} and \eqref{alpha}, respectively. We are going to show that 
the answer to the $3$-dimensional matching problem is yes if and only if the constructed problem \eqref{construct} is feasible.

We first show that if the answer to the $3$-dimensional matching problem is yes, then problem \eqref{construct} is feasible. In fact, if $\left\{(k_x,j_y,l_z)\right\}$ is {a match for the $3$-dimensional matching problem,}
~then a feasible power allocation of problem \eqref{construct} is given by
\begin{equation}\label{power}
  p_{k_x}^n=\left\{\begin{array}{lr}
       3, \quad \text{if}~n=j_y;\\[3pt]
       2, \quad \text{if}~n=l_z;\\[3pt]
       0, \quad \text{if}~n\neq j_y~\text{or}~l_z.
    \end{array}\right.
  \end{equation}
This is because, since $\left\{(k_x,j_y,l_z)\right\}$ is a match for the $3$-dimensional matching problem, 
the above power allocation strategy \eqref{power} is orthogonal to each other (i.e., $p_{k_x}^np_{j_x}^n=0,~\forall~k_x\neq j_x,~k_x,~j_x\in\X,~\forall~n\in\N$).
Furthermore, we have for user $k_x=1_x,2_x,...,K_x,$
\begin{equation*}\begin{array}{rcl}\displaystyle
R_{k_x}&=&\displaystyle\sum_{n\in\N}\log_2\left(1+\frac{\alpha_{k_x,k_x}^np_{k_x}^n}{\eta_{k_x}^n}\right)\quad
        (\text{from}~\eqref{construct})\\[16pt]
&\overset{(a)}{=}&\log_2\left(1+\frac{\alpha_{k_x,k_x}^{j_y}p_{k_x}^{j_y}}{\eta_{k_x}^{j_y}}\right)+\log_2\left(1+\frac{\alpha_{k_x,k_x}^{l_z}p_{k_x}^{l_z}}{\eta_{k_x}^{l_z}}\right)\\[20pt]
&\overset{(b)}{=}&\log_2\left(1+\frac{1*3}{1}\right)+\log_2\left(1+\frac{1*2}{2}\right)\\[10pt]
&=&3,\end{array}\end{equation*} 
where {$(a)$ is due to \eqref{power} and $(b)$ is due to \eqref{eta},~\eqref{alpha}~\text{and}~\eqref{power}.}
So \eqref{power} is a feasible power allocation of problem \eqref{construct}.

On the other hand, we show that if all the constraints in \eqref{construct} are satisfied, the answer to the $3$-dimensional matching problem must be yes. Notice that for any user $k_x\in\K$, if the transmission rate $R_{k_x}\ge 3$, it must transmit on at least two subcarriers, for otherwise if it transmits only on one subcarrier, its maximum transmission rate is $$\log_2\left(1+\frac{1*3}{1}\right)=2<3.$$  Since there are $K$ users and $N=2K$ subcarriers and at most one user is allowed to transmit on each subcarrier (by the OFDMA constraint), the feasibility of \eqref{construct} asks that each user in the network must transmit on \emph{exactly} two subcarries.
By the construction of the parameters of the system, one can verify that the corresponding direct-link channel gains of the user on the two subcarriers must be $1$, for otherwise the transmission rate is at most
  $$\log_2\left(1+\frac{1*3}{1}\right)+\log_2\left(1+\frac{0.25*3}{1}\right)=\log_2 7 < 3.$$
 Furthermore, the fact that all users' transmission rate requirements are satisfied implies that each user transmits on one subcarrier in $\left\{1_y,2_y,...,K_y\right\}$ with noise power $1$ and one in $\left\{1_z,2_z,...,K_z\right\}$ with noise power $2,$ which makes the transmission rate equal to
$$ \log_2\left(1+\frac{1*3}{1}\right)+\log_2\left(1+\frac{1*2}{2}\right) = 3. $$ Otherwise,
\begin{itemize}
  \item [-] if one user transmits on two subcarries in $\left\{1_z,2_z,...,K_z\right\},$ the transmission rate is at most
  $$\log_2\left(1+\frac{1*2}{2}\right)+\log_2\left(1+\frac{1*2}{2}\right)=2<3.$$ 
  \item [-] if one user transmits on two subcarries in $\left\{1_y,2_y,...,K_y\right\}$ both with noise power $1$, then at least one user will transmit on two subcarries in $\left\{1_z,2_z,...,K_z\right\}$ either with noise power $2$ or $3$ and the transmission rate of this user is at most $2.$
 \end{itemize} Therefore, problem \eqref{construct} is feasible if and only if that for all users $k_x~(k_x\in\K)$ there exist $j_y$ and $l_z$ such that \begin{equation*}\label{k1k2}\alpha_{k_x,k_x}^{j_y}=\alpha_{k_x,k_x}^{l_z}=1,~(k_x, j_y)\in {\S}_1,~(k_x, l_z)\in {\S}_2\end{equation*} and
\begin{equation*}\label{union}\left\{1_y,2_y,...,K_y,1_z,2_z,...,K_z\right\}=\bigcup_{k=1}^K\left\{j_y, l_z\right\}.\end{equation*}
According to the construction of ${\S}_1,~{\S}_2$ and \eqref{alpha}, we see that $$\left\{(k_x, j_y, l_z)\right\}_{k=1}^K$$ is {a match for the $3$-dimensional matching problem.}

It is simple to check that the above transformation from the $3$-dimensional matching problem to the feasibility problem \eqref{construct} can be done in polynomial time. Since the $3$-dimensional matching problem is strongly NP-complete, we conclude that checking the feasibility of problem \eqref{construct} is strongly NP-hard. Therefore, the optimization problem \eqref{problem-simp} is also strongly NP-hard. 
%
%
\end{proof}

To illustrate the above proof, we take the following $3$-dimensional matching problem as an example,
$$\X=\left\{1_x,2_x,3_x,4_x\right\},$$$$\Y=\left\{1_y,2_y,3_y,4_y\right\},$$ $$\Z=\left\{1_z,2_z,3_z,4_z\right\},$$ and \begin{align*}\R=&\left\{(1_x,2_y,2_z),(1_x,2_y,4_z),(2_x,1_y,2_z),\right.\\&\left.(2_x,1_y,3_z),(3_x,2_y,2_z),(3_x,4_y,3_z),(4_x,3_y,1_z)\right\}.\nonumber\end{align*} It is simple to check that \begin{equation}\label{matchsolution}\left\{(1_x,2_y,4_z),(2_x,1_y,2_z),(3_x,4_y,3_z),(4_x,3_y,1_z)\right\}\end{equation} is a match for the above given instance of the $3$-dimensional matching problem. 
Based on this $3$-dimensional matching problem, we construct a $4$-user $8$-carrier system with $\K=\left\{1_x,2_x,3_x,4_x\right\}$ and $\N=\left\{1_y,2_y,3_y,4_y,1_z,2_z,3_z,4_z\right\}$.  According to \eqref{s1s2}, we have
$${\S}_1=\left\{(1_x,2_y),(2_x,1_y),(3_x,2_y),(3_x,4_y),(4_x,3_y)\right\}$$ and
\begin{align*}{\S}_2=&\left\{(1_x,2_z),(1_x,4_z),(2_x,2_z),\right.\\&\left.(2_x,3_z),(3_x,2_z),(3_x,3_z),(4_x,1_z)\right\}.\end{align*}
The proof of Lemma \ref{thm-feasi} suggests the following system parameters (cf. \eqref{p}--\eqref{alpha}): all power budgets per subcarrier are
$$P_k^n=3,~k=1_x,2_x,3_x,4_x,~n=1_y,2_y,3_y,4_y,$$$$P_k^n=2,~k=1_x,2_x,3_x,4_x,~n=1_z,2_z,3_z,4_z;$$all noise powers are $\eta_k^n=3~(k\in\K,n\in\N)$ except
$$\eta_{1_x}^{2_y}=\eta_{2_x}^{1_y}=\eta_{3_x}^{2_y}=\eta_{3_x}^{4_y}=\eta_{4_x}^{3_y}=1,$$$$\eta_{1_x}^{2_z}=\eta_{1_x}^{4_z}=\eta_{2_x}^{2_z}=\eta_{2_x}^{3_z}=\eta_{3_x}^{2_z}=\eta_{3_x}^{3_z}=\eta_{4_x}^{1_z}=2;$$ and all direct-link channel gains are $\alpha_{k,k}^n=0.25~(k\in\K,n\in\N)$ except
\begin{equation*}
  \begin{array}{cl}
    \alpha_{1_x,1_x}^{2_y}=\alpha_{1_x,1_x}^{2_z}=\alpha_{1_x,1_x}^{4_z}=\alpha_{2_x,2_x}^{1_y}=\alpha_{2_x,2_x}^{2_z}=\alpha_{2_x,2_x}^{3_z}=1,\\[3pt]
    \alpha_{3_x,3_x}^{2_y}=\alpha_{3_x,3_x}^{2_z}=\alpha_{3_x,3_x}^{4_y}=\alpha_{3_x,3_x}^{3_z}=\alpha_{4_x,4_x}^{3_y}=\alpha_{4_x,4_x}^{1_z}=1.
  \end{array}
\end{equation*}

In this example,
based on the {match} \eqref{matchsolution}, we can construct an OFDMA solution to the corresponding feasibility check problem \eqref{construct}, i.e.,
\begin{equation*}\label{conssolution}
\begin{array}{rl}
\left(p_{1_x}^{1_y},p_{1_x}^{2_y},p_{1_x}^{3_y},p_{1_x}^{4_y},p_{1_x}^{1_z},p_{1_x}^{2_z},p_{1_x}^{3_z},p_{1_x}^{4_z}\right)=(0,3,0,0,0,0,0,2),\\[5pt]
\left(p_{2_x}^{1_y},p_{2_x}^{2_y},p_{2_x}^{3_y},p_{2_x}^{4_y},p_{2_x}^{1_z},p_{2_x}^{2_z},p_{2_x}^{3_z},p_{2_x}^{4_z}\right)=(3,0,0,0,0,2,0,0),\\[5pt]
\left(p_{3_x}^{1_y},p_{3_x}^{2_y},p_{3_x}^{3_y},p_{3_x}^{4_y},p_{3_x}^{1_z},p_{3_x}^{2_z},p_{3_x}^{3_z},p_{3_x}^{4_z}\right)=(0,0,0,3,0,0,2,0),\\[5pt]
\left(p_{4_x}^{1_y},p_{4_x}^{2_y},p_{4_x}^{3_y},p_{4_x}^{4_y},p_{4_x}^{1_z},p_{4_x}^{2_z},p_{4_x}^{3_z},p_{4_x}^{4_z}\right)=(0,0,3,0,2,0,0,0).
\end{array}
\end{equation*}

On the other hand, to look for a feasible solution of problem \eqref{construct}, we have to make each user transmit on two subcarriers (one with noise power $1$ and the other with noise power $2$) and the direct-link channel gains of the user on the corresponding two subcarriers be $1.$  Notice that $\eta_{1_x}^{2_y}=\eta_{2_x}^{1_y}=\eta_{3_x}^{4_y}=\eta_{4_x}^{3_y}=1,$~$\eta_{1_x}^{4_z}=\eta_{2_x}^{2_z}=\eta_{3_x}^{3_z}=\eta_{4_x}^{1_z}=2,$ and $\alpha_{1_x,1_x}^{2_y}=\alpha_{1_x,1_x}^{4_z}=\alpha_{2_x,2_x}^{1_y}=\alpha_{2_x,2_x}^{2_z}=\alpha_{3_x,3_x}^{4_y}=\alpha_{3_x,3_x}^{3_z}=\alpha_{4_x,4_x}^{3_y}=\alpha_{4_x,4_x}^{1_z}=1.$
We can ask user $1_x$ to transmit on subcarriers $2_y$ and $4_z,$ user $2_x$ to transmit on subcarriers $1_y$ and $2_z,$ user $3_x$ to transmit on subcarriers $4_y$ and $3_z,$ and user $4_x$ to transmit on subcarriers $3_y$ and $1_z$, respectively. Consequently,  $\left\{(1_x,2_y,4_z),(2_x,1_y,2_z),(3_x,4_y,3_z),(4_x,3_y,1_z)\right\}$ is a {match for the given instance of} the $3$-dimensional matching problem.


Lemma \ref{thm-feasi} shows that checking the feasibility of problem \eqref{problem-simp} is strongly NP-hard when $c=2.$
 Based on this basic result, we can further prove that it is strongly NP-hard to check the feasibility of problem \eqref{problem-simp} when $N/K=c$ provided that $c$ is a strictly greater than one constant. We summarize this result as Theorem \ref{thm-feasi2}, and relegate its proof to Appendix \ref{app-thm-feasi2}.
\begin{dingli}\label{thm-feasi2}
  Given any constant $c>1,$ checking the feasibility of problem \eqref{problem-simp} is strongly NP-hard when $N/K=c$. Thus, problem \eqref{problem-simp} itself is also strongly NP-hard.
\end{dingli}

{{Remark} 1: Problem \eqref{problem-simp} remains strongly NP-hard if the per-subcarrier power budget constraints $P_k^n\geq p_k^n\geq 0~(k\in\K,~n\in\N)$ there are replaced by the total power constraints $\sum_{n\in\N} p_k^n\leq P_k~(k\in\K)$ or $\sum_{k\in\K}\sum_{n\in\N}p_k^n\leq P.$ By setting $P_k=5~(k\in\K)$ or $P=5K$ and using the same argument as in the proof of Lemma \ref{thm-feasi} and Theorem \ref{thm-feasi2}, the strong NP-hardness of the corresponding feasibility problems can be shown. In fact, all strong NP-hardness results of problems \eqref{problem-simp} and \eqref{utility} in this paper also hold true for problems with either of the above two total power constraints.}

{Remark 2: Another extension of problem \eqref{problem-simp} is the so-called joint subcarrier and bit allocation problem \cite{JSAC,newh}. The goal of the joint subcarrier and bit allocation problem is to allocate subcarriers to users and at the same time allocate transmission bits to each user-subcarrier pair such that the total transmission power is minimized and the OFDMA constraints and all users' transmission requirements are satisfied. Mathematically, the problem can be formulated as
\begin{equation}\label{subcarrier-bit}
\begin{array}{cl}
\displaystyle \min_{\left\{p_k^n\right\},\,\left\{r_k^n\right\}} & \displaystyle \sum_{k\in\K}\sum_{n\in\N}p_k^{n}\\[15pt]
\text{s.t.} & \displaystyle \sum_{n\in\N} r_k^n\geq \gamma_k,~k\in\K,\\[15pt]
& \displaystyle p_k^n={(2^{r_k^n}-1)\eta_k^n}/{\alpha_{k,k}^n},~k\in\K,~n\in\N, \\[5pt]
    & r_k^n\in\left\{r_1,r_2,...,r_m,0\right\},~k\in\K,~n\in\N,\\[5pt]
    & p_k^np_j^n=0,~\forall~j\neq k,{~k,\,j\in\K,}~n\in\N.
    \end{array}
\end{equation}The second constraint in \eqref{subcarrier-bit} says that ${(2^{r_k^n}-1)\eta_k^n}/{\alpha_{k,k}^n}$ (cf. \eqref{rk}) is necessary for user $k$ transmitting on subcarrier $n$ to achieve transmission rate $r_k^n$, and the third constraint in \eqref{subcarrier-bit} enforces $r_k^n$ to take values in the possible transmission rate set $\left\{r_1,r_2,...,r_m,0\right\}.$

Given a total power budget $P>0$, the decision version of problem \eqref{subcarrier-bit} is to ask
whether there exists a feasible power and bit allocation strategy such that the optimal value of \eqref{subcarrier-bit} is less than
or equal to $P.$ 
It was shown in \cite{newh} that the decision version of problem \eqref{subcarrier-bit} is strongly NP-hard when $m\geq 2$. The proof is based on a polynomial time transformation from the scheduling problem\cite{Complexitybook}.
In fact, by setting $P=5K, m=2, r_1=2, r_2=1$ and using the same argument as in Lemma 3.1 and Theorem 3.1, we can also show the strong NP-hardness of
the decision version of problem \eqref{subcarrier-bit}.

}


\subsection{Strong NP-Hardness of Utility Maximization Problem \eqref{utility} when $N/K=c>1$}\label{sub-iii}
{In this subsection, we study the computational complexity of} the {system} utility maximization problem \eqref{utility} for the multi-user OFDMA system, where the utility is one of the four system utility functions introduced in Section \ref{sec-model}. Theorem \ref{thm-uti2}
establishes the strong NP-hardness of the problem for the general case $c>1$. The proof of Theorem \ref{thm-uti2} can be found in Appendix \ref{app-thm-uti2}.
%
%
%

\begin{dingli}\label{thm-uti2}Given any constant $c>1,$ the system utility maximization problem \eqref{utility} with $H=H_1, H_2, H_3,$ or $H_4$ is strongly NP-hard when $N/K=c.$
\end{dingli}

 {Remark 3:} The result in Theorem \ref{thm-uti2} is different from the one in \cite{hayashi}. It was shown in \cite{hayashi} that the sum-rate utility maximization problem \eqref{utility} (i.e., $H=H_1$) is NP-hard when the number of users is equal to $2.$ The proof is based on a polynomial time transformation from the equipartition problem \cite{Complexitybook}, which is known to be NP-complete but not strongly NP-complete. Theorem \ref{thm-uti2} shows that the sum-rate utility maximization problem is \emph{strongly} NP-hard. The other three utility functions ($H=H_2,H_3,$ or $H_4$) are not considered in \cite{hayashi}. The reference \cite{complexity} proved the NP-hardness of the three utility maximization problems ($H=H_2,H_3,$ or $H_4$) in the two-user case by establishing a polynomial time transformation from the equipartition problem. However, the proof of \cite{complexity} is nonrigorous, since the given equipartition problem has a yes answer does not imply that the transmission rate of the two users (at the solution of the corresponding utility maximization problem) is equal to each other, and vice versa. The complexity status of problem \eqref{utility} with the proportional fairness utility, the harmonic mean utility, and the min-rate utility remains unknown when $K\geq 2$ is fixed.

{In this section, we have shown that both the total power minimization problem \eqref{problem-simp} and the system utility maximization problem \eqref{utility} are strongly NP-hard. 
The basic idea of the proof is establishing a polynomial time transformation from the $3$-dimensional matching problem to the decision version of problems \eqref{problem-simp} and \eqref{utility}. The complexity result suggests that there are not polynomial time algorithms which can solve problems \eqref{problem-simp} and \eqref{utility} to global optimality (unless P$=$NP). Therefore, one should abandon efforts to find globally optimal subcarrier and power allocation strategy for problems \eqref{problem-simp} and \eqref{utility}, and determining an approximately optimal subcarrier and power allocation strategy is more realistic in practice.
}

\section{Easy Cases}\label{sec-easy}

 In this section, we identify some easy cases when problem \eqref{problem-simp} or problem \eqref{utility} can be solved in polynomial time.
Before doing this, we introduce a concept called \emph{strong} polynomial time algorithm.
A problem is said to admit a strong polynomial time algorithm if there exists an algorithm satisfying the following two conditions:
 \begin{itemize}
   \item [a)] the complexity of the algorithm (when applied to solve the problem) depends only on the dimension of the problem and is a polynomial function of the dimension;
   \item [b)] the algorithm solves the problem to global optimality (not just $\epsilon$-global optimality).
 \end{itemize} We remark that so far we do not know whether there exists a strong polynomial time algorithm to solve the general linear programming.
 When the interior-point algorithm is applied to solve the linear programming, it is only guaranteed to return an $\epsilon$-optimal solution in polynomial time and the complexity of the interior-point algorithm depends on the factor $\log_2(1/\epsilon)$\cite{com3}. The best complexity results of solving the linear programming is still related to the condition number of the constraint matrix \cite{VavasisYe}.

{In the following subsections, we identify four (strong) polynomial time solvable subclasses of problems \eqref{problem-simp} and \eqref{utility}. More specifically, we first show that both problem \eqref{problem-simp} and problem \eqref{utility} are (strongly) polynomial time solvable when there is only one user in the system; see Subsection \ref{Sub-K=1} and Subsection \ref{Subsec-K=1}, respectively. The (extended) ``water-filling'' technique plays a fundamental role in proving the polynomial time complexity. Then, we show in Subsection \ref{SubsecN=K} that problem \eqref{problem-simp} is strongly polynomial time solvable when the number of subcarriers is equal to the number of users. In this case, we can reformulate problem \eqref{problem-simp} as an assignment problem for a complete bipartite graph, which can be solved in strong polynomial time. Finally, we show the polynomial time complexity of problem \eqref{utility} with sum-rate utility without the total power constraint by transforming it into the polynomial time solvable Hitchcock problem in Subsection \ref{Subsec-Nopower}.
}

 \subsection{Polynomial Time Solvability of Problem \eqref{problem-simp} when $K=1$}\label{Sub-K=1}
  When there is only one user (i.e., $K=1$), problem \eqref{problem-simp} becomes
 \begin{equation}\label{K=1}
 \begin{array}{rl}
\displaystyle {P^*}=\min_{\{p^n\}} &\displaystyle \sum_{n\in\N}p^{n}\\[5pt]
\text{s.t.} & \displaystyle\sum_{n\in\N}\log_2\left(1+\frac{\alpha^np^{n}}{\eta^n}\right)\geq \gamma, \\[15pt]
    & P^n\geq p^n\geq 0,~n\in\N.
    \end{array}
\end{equation}
We claim that solving problem \eqref{K=1} is equivalent to finding a minimal $P$ such that the {optimal} value of problem
\begin{equation}\label{sub2}
\begin{array}{cl}
\displaystyle \max_{\{p^n\}} & \displaystyle \sum_{n\in\N}\log_2\left(1+\frac{\alpha^np^{n}}{\eta^n}\right)\\[15pt]
\text{s.t.} & \displaystyle\sum_{n\in\N}p^n\leq P, \\[15pt]
    &  P^n\geq p^n\geq 0,~n\in\N
    \end{array}\end{equation} is equal to $\gamma.$ This is an important observation towards obtaining the closed-form solution of problem \eqref{K=1}. In fact, if the optimal value of problem \eqref{K=1} is {$P^*$}, then the optimal value of problem \eqref{sub2} {with $P=P^*$} is $\gamma;$ {and vice versa.}
    In more details, for any fixed $P>0$, problem \eqref{sub2} is strictly convex with respect to $\{p^n\}$ and hence has a unique solution. Therefore,
    the set defined by
    \begin{equation*}\left\{\!\!\!\!\!
\begin{array}{cl} & \displaystyle \sum_{n\in\N}\log_2\left(1+\frac{\alpha^np^{n}}{\eta^n}\right)\geq \gamma,\\[15pt]
     & \displaystyle\sum_{n\in\N}p^n\leq {P^*}, \\[15pt]
    &  P^n\geq p^n\geq 0,~n\in\N
    \end{array}\right.\end{equation*} contains only one point, which must be the solution of problem \eqref{sub2} {with $P=P^*$}. Hence, problems \eqref{K=1} and \eqref{sub2} share the same solution.

The
 solution to problem \eqref{K=1} or problem \eqref{sub2} is given by the following extended water-filling solution
\begin{equation}\label{pkn2}p^n=\min\left\{P^n, \left(\frac{1}{\lambda}-\frac{\eta^n}{\alpha^n}\right)_+\right\},~n\in\N,\end{equation}where \emph{$\lambda$ is chosen such that the objective value of problem \eqref{sub2} is equal to $\gamma,$} and $(x)_+=\max\left\{x,0\right\}.$ After obtaining the optimal $\lambda$, the optimal value $P^*$ of problem \eqref{K=1} is given by
$$P^*=\sum_{n\in\N} p^n=\sum_{n\in\N}\min\left\{P^n, \left(\frac{1}{\lambda}-\frac{\eta^n}{\alpha^n}\right)_+\right\}.$$ For completeness, we show in Appendix \ref{app-2} that \eqref{pkn2} indeed is the solution to problem \eqref{sub2}
and $\lambda$ is actually the Lagrangian multiplier corresponding to the constraint $\sum_{n\in\N}p^n\leq P$.

 We point out that the water-filling solution \eqref{pkn2} extends the conventional water-filling solution
  \begin{equation}\label{waterfilling}p^n=\left(\frac{1}{\lambda}-\frac{\eta^n}{\alpha^n}\right)_+,~n\in\N\end{equation}
  in \cite{cover} in the following two respects. First, the conventional water-filling solution \eqref{waterfilling} solves the power control problem \eqref{sub2} without the power budget constraints per subcarrier, while the power control problem \eqref{sub2} not only involves the total power constraint but also involves the power budget constraints per subcarrier. Second, the parameter $\lambda$ in the conventional water-filling solution \eqref{waterfilling} is chosen such that $\sum_{n\in\N}p^n=P,$ while the parameter $\lambda$ in \eqref{pkn2} is chosen such that $$\sum_{n\in\N}\log_2\left(1+\frac{\alpha^np^{n}}{\eta^n}\right)=\gamma.$$

The only left problem now is to find $\tau:=1/\lambda$ in \eqref{pkn2} such that the objective value of problem \eqref{sub2} is equal to $\gamma.$
A natural way to find the desired $\tau^*$ is to perform a binary search on $\tau,$ since the objective function of \eqref{sub2} is an increasing function with respect to $\tau.$ As is known, the efficiency of the binary search depends on the initial search interval of $\tau.$  To derive a good lower and upper bound, we first order the sequence
\begin{equation} \left\{b^n\right\}_{n=1}^{2N}:=\left\{\frac{\eta^n}{\alpha^n}, \frac{\eta^n}{\alpha^n}+P^n\right\}_{n=1}^N,
 \label{add1}\end{equation} and without loss of generality, suppose that
 \begin{equation} b^1\leq b^2\leq\cdots\leq b^{2N}.  \label{add2}\end{equation} Notice that it takes $O(2N\log_2(2N))$ operations to order $\left\{b^n\right\}_{n=1}^{2N}.$ Then we calculate the objective values of problem \eqref{sub2} at $\tau=b^n~(n=1,2,...,2N),$ denoting {them} by $\left\{v^n\right\}_{n=1}^{2N}$. It follows from the monotonicity {of $\left\{b^n\right\}$} that $$v^1\leq v^2\leq\cdots\leq v^{2N}.$$
If there exists an index $n^*$ such that $v^{n^*}=\gamma,$ then $\tau^*=b^{n^*}.$ Otherwise, we have $v^{n^*}< \gamma <v^{n^*+1}$ for some $n^*$ and hence $\tau^*\in (b^{n^*},b^{n^*+1}).$  In this case, we can start the binary search from $(b^{n^*},b^{n^*+1})$ and it takes at most $\log_2\left((b^{n^*+1}-b^{n^*})/\epsilon\right)$ iterations to obtain an $\epsilon$-optimal $\tau.$
{Assume that $$\max_{n}\left\{\frac{\eta^n}{\alpha^n}+P^n\right\}\leq R,$$ where $R$ is a sufficiently large constant\footnote{This is a standard assumption in the complexity analysis of the interior-point methods for solving convex conic programming \cite{com3}.}. Then, we can conclude that problem \eqref{K=1} is polynomial time solvable and the worst case complexity of solving it is $$O(\log_2(R/\epsilon)+N\log_2\left(N\right)).$$
}

{Remark 4: The complexity status of the total power minimization problem \eqref{problem-simp} remains unknown when $K\geq 2$ is fixed.}

%
\subsection{Strong Polynomial Time Solvability of Problem \eqref{utility} when $K=1$}\label{Subsec-K=1}

 If the system has only one user, all the four system utility functions coincide and problem \eqref{utility} becomes
  \begin{equation}\label{uK=1}
\begin{array}{cl}
\displaystyle \max_{\{p^n\}} & \displaystyle \sum_{n\in\N}\log_2\left(1+\frac{\alpha^np^{n}}{\eta^n}\right)\\[15pt]
\text{s.t.} & \displaystyle \sum_{n\in\N}p^n\leq P, \\[15pt]
    &  P^n\geq p^n\geq 0,~n\in\N.
    \end{array}\end{equation} 
The main difference between problem \eqref{uK=1} and problem \eqref{sub2} lies in that, $P$ is a given constant in problem \eqref{uK=1}, but is an unknown parameter in problem \eqref{sub2}. This feature of $P$ in problem \eqref{uK=1} enables us to design a strong polynomial time algorithm for the problem. {Without loss of generality, we assume that the parameters in problem \eqref{uK=1} satisfy $P\leq\sum_{n\in\N}P^n;$ otherwise, the solution to problem \eqref{uK=1} is $p^n=P^n$ for all $n\in\N.$}

Notice that the solution to problem \eqref{uK=1} is given by the extended water-filling solution \eqref{pkn2}, where $\tau:=1/\lambda$ is chosen such that $\sum_{n\in\N}p^n = P.$ The desired {${\tau^*}$} here can be found by solving a univariate linear equation. In fact,
in a similar fashion {as in Subsection \ref{Sub-K=1}}, we can order $\left\{b^n\right\}_{n=1}^{2N}$ in \eqref{add1} and assume that \eqref{add2} holds.  Then we calculate the total transmission power $\sum_{n\in\N}p^n$ at $\tau=b^n~(n=1,2,...,2N)$, denoting {them} by $\left\{u^n\right\}_{n=1}^{2N}$. It follows from the monotonicity {of $\left\{b^n\right\}$} that
$$u^1\leq u^2\leq \cdots\leq u^{2N}.$$
If there is an index ${n^*}$ such that $P=u^{n^*}$, then {${\tau}^*=b^{n^*}$}. Otherwise, we have that
$u^{n^*}< P <u^{n^*+1}$ for some $1\leq n^*< N$.   Then for each $n=1,2,...,N,$ we have
    \begin{align*}\label{pkn3}
  &\min\left\{P^n, \left({\tau}-\frac{\eta^n}{\alpha^n}\right)_+\right\} \\[5pt]
  =&\left\{\begin{array}{ll}
       P^n, \quad \text{if}~\frac{\eta^n}{\alpha^n}+P^n\leq b^{n^*};\\[10pt]
       {\tau}-\frac{\eta^n}{\alpha^n}, \quad \text{if}~\frac{\eta^n}{\alpha^n}< b^{n^*}<\frac{\eta^n}{\alpha^n}+P^n;\\[10pt]
       0, \quad \text{if}~\frac{\eta^n}{\alpha^n}\geq b^{n^*}.
    \end{array}\right.
  \end{align*}  Therefore, the problem of finding the desired ${\tau}^*$ reduces to solve a univariate linear equation in terms of $\tau,$ and the desired ${\tau}^*$ is obtained in a closed form. From the above discussion, we see that the complexity of finding the desired ${\tau}^*$ is $O(N\log_2(N)).$

{Remark 5:} The reference \cite{complexity} has addressed problem \eqref{uK=1}, but without power budget constraints per subcarrier. {The reference \cite{complexity} has also shown that the sum-rate maximization problem \eqref{utility} is NP-hard when $K\geq2.$}

\subsection{Strong Polynomial Time Solvability of Problem \eqref{problem-simp} when $N=K$}\label{SubsecN=K}
When the number of users is equal to the number of subcarriers (i.e., $N=K$), we transform problem \eqref{problem-simp} into an assignment problem for a complete bipartite graph with $2K$ nodes in polynomial time. Specifically, we construct the bipartite graph ${{G}}=(V,E)$ with
\begin{itemize}
  \item the node set  $$V=\left\{1,2,...,K,1',2',...,K'\right\},$$ where {$\K=\left\{1,2,...,K\right\}$ and $\N=\left\{1',2',...,K'\right\}$} correspond to the set of users and the set of subcarriers, respectively;
  \item the edge set $$E=\left\{(k,n)~|~k\in\left\{1,2,...,K\right\},~n\in\left\{1',2',...,K'\right\}\right\}$$ with the weight $w_{k}^n$ of edge $(k,n)$
  \begin{equation}\label{weight}
  w_{k}^n=\left\{\begin{array}{lr}
       \frac{\left(2^{\gamma_k}-1\right)\eta_k^n}{\alpha_{k,k}^n}, \quad \text{if}~\log_2\left(1+\frac{\alpha_{k,k}^nP_k^n}{\eta_{k}^n}\right)\geq \gamma_k;\\[15pt]
       \displaystyle\sum_{k\in\K}\sum_{n\in\N}P_k^n, \quad \text{otherwise}.
    \end{array}\right.
  \end{equation}
\end{itemize}
Therefore, problem \eqref{problem-simp} can be equivalently reformulated as
\begin{equation}\label{assignment}
  \begin{array}{rl}
    \displaystyle \left\{\left(x_k^n\right)^*\right\}=\arg\min_{\left\{x_{k}^n\right\}} & \displaystyle \sum_{k\in\K}\sum_{n\in\N} w_{k}^nx_{k}^n \\[15pt]
\text{s.t.}
    & \displaystyle \sum_{k\in\K} x_{k}^n=1,~n\in\N, \\[12pt]
    & \displaystyle \sum_{n\in\N} x_{k}^n=1,~k\in\K,\\[12pt]
    & x_{k}^n\in\left\{0,1\right\},~k\in\K,~n\in\N.
  \end{array}
\end{equation} In the above problem, the binary variable $x_{k}^n$ is equal to $1$ if user $k$ transmits power on subcarrier $n$ and $0$ otherwise.
 The first constraint $\sum_{k\in\K} x_{k}^n=1~(n\in\N)$ stands for the OFDMA constraint, which requires that at most one user is allowed to transmit power on each subcarrier. The second constraint $\sum_{n\in\N} x_{k}^n=1~(k\in\K)$ requires that each user must transmit on one subcarrier to satisfy its specified transmission rate {requirement}.

 From \cite[Theorem 11.1]{combinatorial}, we know that the Hungarian method solves problem \eqref{assignment} to global optimality in $O(K^3)$ operations. The Hungarian method is in essence a primal-dual simplex algorithm for solving the linear program
 \begin{equation}\label{assignment-relax}
  \begin{array}{rl}
    \displaystyle \left\{\left(x_k^n\right)^*\right\}=\arg\min_{\left\{x_{k}^n\right\}} & \displaystyle \sum_{k\in\K}\sum_{n\in\N} w_{k}^nx_{k}^n \\[15pt]
\text{s.t.}
    & \displaystyle \sum_{k\in\K} x_{k}^n=1,~n\in\N, \\[12pt]
    & \displaystyle \sum_{n\in\N} x_{k}^n=1,~k\in\K,\\[12pt]
    & 0\leq x_{k}^n\leq 1,~k\in\K,~n\in\N,
  \end{array}
\end{equation}which is a relaxation of problem \eqref{assignment} by replacing $x_{k}^n\in\left\{0,1\right\}$ with $0\leq x_{k}^n\leq 1.$ In general, the primal-dual simplex method is not a polynomial time algorithm. When it is used to solve problem \eqref{assignment-relax}, however, it can return the \emph{global optimal integer solution} of problem \eqref{assignment-relax} in $O(K^3)$ operations \cite{combinatorial}.


 We remark that if the optimal value ${P^*}$ of problem \eqref{assignment} satisfies $$P^*<\sum_{k\in\K}\sum_{n\in\N}P_k^n,$$
 the optimal solution to \eqref{problem-simp} is
$$(p_k^n)^*=(x_{k}^n)^*\frac{\left(2^{\gamma_k}-1\right)\eta_k}{\alpha_{k,k}^n},~k\in\K,~n\in\N;$$ otherwise, we will have that $$P^*\geq\sum_{k\in\K}\sum_{n\in\N}P_k^n,$$
which indicates that problem \eqref{problem-simp} is infeasible.

{Remark 6:} In fact, problem \eqref{problem-simp} is also polynomial time solvable when $N=K+C,$ where $C\geq 1$ is a given constant integer. Specifically, we can first partition $K+C$ subcarriers into $K$ nonempty subsets $\left\{\N_n\right\}_{n=1'}^{K'}.$ Denote the number of ways to partition $K+C$ subcarriers into $K$ nonempty subsets by $S(K+C,K).$ We show in Appendix \ref{app-s} that $S(K+C,K)$ is upper bounded by $(K+C)^{2C}.$
%

After the partition of the subcarriers, the problem reduces to the case $N=K.$ 
  The only difference here is that the parameter $w_k^n~(k=1,2,...,K,~n=1',2',...,K')$ is set to be 
\begin{equation}\label{power-allocation2}
 \begin{array}{rl}
\displaystyle w_k^n=\min_{\left\{p_k^l\right\}_{l\in{\cal N}_n}} & \displaystyle \sum_{l\in {\cal N}_n}p_k^{l}\\[12pt]
\text{s.t.} & \displaystyle \sum_{l\in\N_n}\log_2\left(1+\frac{\alpha_{k,k}^lp_k^{l}}{\eta_k^l}\right)\geq \gamma_{k}, \\[15pt]
    & P_k^l\geq p_k^l\geq 0,~l\in {\cal N}_n\\
    \end{array} \end{equation} if the above problem is feasible; otherwise 
    $$w_k^n=\displaystyle\sum_{k=1}^K\sum_{l=1}^{N}P_k^l.$$ We know from Subsection \ref{Sub-K=1} that problem \eqref{power-allocation2} is polynomial time solvable {under a mild assumption.} Actually, if problem \eqref{power-allocation2} is feasible and $\N_n=\left\{l\right\}$ only contains a singleton, we have $$w_k^n=\frac{\left(2^{\gamma_k}-1\right)\eta_k^{l}}{\alpha_{k,k}^l},$$ which is the same as the one in \eqref{weight}.
    Then, for a given partition $p=1,2,...,S(K+C,K),$ we can use the Hungarian method to solve problem \eqref{assignment} in $O(K^3)$ operations, and assume the optimal value to be $v_p.$
    Therefore, the optimal value of the original problem is $$\min_{p=1,2,...,S(K+C,K)}\left\{v_p\right\}$$ if it is strictly less than $\sum_{k=1}^K\sum_{l=1}^{N}P_k^l;$ otherwise the original problem is infeasible. It follows from the above analysis that problem \eqref{problem-simp} with $N=K+C$ can be solved in $O(K^{2C+3})$ operations.

\subsection{Polynomial Time Solvability of Sum-Rate Maximization Problem \eqref{utility} without Total Power Constraint}\label{Subsec-Nopower}
If we drop the total power constraint $\sum_{n\in\N}p_k^{n}\leq P_{k}~(k\in\K),$ 
then problem \eqref{utility} with $H=H_1$ becomes 
\begin{equation}\label{utility2}
 \begin{array}{cl}
\displaystyle \max_{\{p_k^n\}} & \displaystyle H_1(R_1,R_2,...,R_K) \\
\text{s.t.} 
    & P_k^n\geq p_k^n\geq 0,~k\in\K,~n\in\N,\\[5pt]
    & p_k^np_j^n=0,~\forall~j\neq k,{~k,\,j\in\K,}~n\in\N.
    \end{array}
\end{equation}
Problem \eqref{utility2} is always feasible as $p_k^n=0$ (for all $k$ and $n$) is a certificate for the feasibility.
Further, since problem \eqref{utility2} does not involve the total power constraint for each user, its optimal solution $p_k^n~(k\in\K,~n\in\N)$ is either $0$ or $P_k^n.$ To solve problem \eqref{utility2} in polynomial time, we consider
 transforming it into the polynomial time solvable Hitchcock problem \cite{combinatorial} (also known as transportation problem).

 {\bf Hitchcock problem}. \ Suppose there are $K$ sources of some commodity, each with a supply of $a_k\in\mathbb{R}^+~(k=1,2,...,K)$ units, and $N$ terminals, each of which has a demand of $d_n\in\mathbb{R}^+~(n=1,2,...N)$ units. Suppose that the unit cost of transporting the commodity from source $k$ to terminal $n$ is $c_k^n\in\mathbb{R}^+~(k=1,2,...,K,~n=1,2,...,N).$ The problem is how to satisfy the demands at a minimal cost?

 In fact, by setting $$a_k=N,~k=1,2,...,K;$$ $$d_n=1,~n=1,2,...N,~d_{N+1}=(K-1)N; $$
 $$c_k^n=\bar c-\log_2\left(1+\frac{\alpha_{k,k}^nP_k^n}{\eta_k^n}\right)\geq 0,~k=1,2,...,K,~n=1,2,...,N;$$$$c_k^{N+1}=0,~k=1,2,...,K,$$ where $\bar c=\max_{k,n}\left\{\log_2\left(1+\frac{\alpha_{k,k}^nP_k^n}{\eta_k^n}\right)\right\},$ we can see that
  problem \eqref{utility2} is equivalent to the following Hitchcock problem
\begin{equation}\label{hitch}
 \begin{array}{cl}
\displaystyle \min_{\left\{x_{k}^n\right\}} & \displaystyle \sum_{k=1}^K\sum_{n=1}^{N+1}c_k^nx_k^n \\[13pt]
\text{s.t.} & \sum_{k=1}^K x_k^n=1,~n=1,2,...,N,\\[5pt]
            & \sum_{k=1}^K x_k^{N+1}=(K-1)N,\\[5pt]
    & \sum_{n=1}^N x_k^n=N,~k=1,2,...,K,\\[5pt]
            & x_k^n=\left\{0,1\right\},~k=1,2,...,K,~n=1,2,...,N.
    \end{array}
\end{equation} In the above problem, the binary variable $x_{k}^n=1~(k=1,2,...,K,\,n=1,2,...,N)$ if user $k$ transmits full power $P_k^n$ on subcarrier $n$ and
$x_{k}^n=0$ if user $k$ does not transmit any power on subcarrier $n$.
 The constraint $\sum_{k=1}^K x_{k}^n=1$ and the constraint $x_k^n=\left\{0,1\right\}~(k=1,2,...,K)$ implies that at most one user is allowed to transmit power on each subcarrier $n$. However, one user is allowed to transmit on multiple subcarriers. The variables $x_k^{N+1}~(k=1,2,...,K)$ are {auxiliary} dummy variables, since $c_k^{N+1}=0$ for all $k=1,2,...,K.$
Moreover, we know from \cite[Theorem 13.3 and its Corollary]{combinatorial} that problem \eqref{hitch} is equivalent to the linear program 
 \begin{equation}\label{hitch2}
 \begin{array}{rl}
\displaystyle \min_{\left\{x_{k}^n\right\}} & \displaystyle \sum_{k=1}^K\sum_{n=1}^{N+1}c_k^nx_k^n \\[13pt]
\text{s.t.} & \sum_{k=1}^K x_k^n=1,~n=1,2,...,N,\\[5pt]
            & \sum_{k=1}^K x_k^{N+1}=(K-1)N,\\[5pt]
    & \sum_{n=1}^N x_k^n=N,~k=1,2,...,K,\\[5pt]
            & 0\leq x_k^n\leq 1,~k=1,2,...,K,~n=1,2,...,N.
    \end{array}
\end{equation}The equivalence between problems \eqref{hitch} and \eqref{hitch2} is because the linear equality constraint in \eqref{hitch2} satisfies the so-called \emph{totally unimodular} property\cite{combinatorial}, and thus all the vertices of the feasible set of problem \eqref{hitch2} are integer. Since the linear program \eqref{hitch2} is polynomial time solvable, the sum-rate maximization problem \eqref{utility2} is polynomial time solvable. Further, {suppose} $\left\{\left(x_k^n\right)^*\right\}$ is the optimal solution to problem \eqref{hitch2}, then the optimal solution $\left\{\left(p_k^n\right)^*\right\}$ to problem \eqref{utility2} is
 $$(p_k^n)^*=(x_k^n)^*P_k^n,~k=1,2,...,K,~n=1,2,...,N.$$ {It is worthwhile remarking} that the so-called $\alpha\beta$-algorithm \cite[Section 7.4]{combinatorial} can efficiently solve problem \eqref{hitch2}, although it is not a polynomial time algorithm.

{Remark 7:} In a similar fashion, we can show that the \emph{weighted} sum-rate maximization problem
\begin{equation}\label{utility3}
 \begin{array}{cl}
\displaystyle \max_{\{p_k^n\}} & \displaystyle \sum_{k=1}^K w_kR_k \\[15pt]
\text{s.t.} 
    & P_k^n\geq p_k^n\geq 0,~k\in\K,~n\in\N,\\[5pt]
    & p_k^np_j^n=0,~\forall~j\neq k,{~k,\,j\in\K,}~n\in\N
    \end{array}
\end{equation}is also polynomial time solvable, where $w_k\geq 0~(k=1,2,...,K)$ are non-negative weights. Again, we can transform problem \eqref{utility3} into the Hitchcock problem in polynomial time. The corresponding parameters of the Hitchcock problem are the same as before, except
$$c_k^n=\bar c-w_k\log_2\left(1+\frac{\alpha_{k,k}^nP_k^n}{\eta_k^n}\right)\geq 0$$ for $k=1,2,...,K,~n=1,2,...,N,$ and
$$\displaystyle \bar c=\max_{k,n}\left\{w_k\log_2\left(1+\frac{\alpha_{k,k}^nP_k^n}{\eta_k^n}\right)\right\}.$$

{We summarize the main results in this section as the following theorem.}
\begin{dingli}The following statements are true.
\begin{itemize}
\item [-] Problem \eqref{problem-simp} is strongly polynomial time solvable when the number of users is equal to the number of subcarriers.
\item [-] Problem \eqref{problem-simp} is polynomial time solvable when there is only one user in the system.
 \item [-] Problem \eqref{utility} is strongly polynomial time solvable when there is only one user in the system.
 \item [-] The (weighted) sum-rate maximization problem \eqref{utility} without the total power constraint is polynomial time solvable.
\end{itemize}

\end{dingli}

{In this section, we have identified four subclasses of problems \eqref{problem-simp} and \eqref{utility} which are (strongly) polynomial time solvable. By doing so, we successfully pick a subset of computationally tractable problems within the general class of strongly NP-hard joint subcarrier and power allocation problems.
}

\section{Concluding Remarks}\label{sec-conclusion}
Dynamic allocation of subcarrier and power resources in accordance with channel and traffic load changes can significantly improve the network throughput and spectral efficiency of the multi-user multi-carrier communication system where a number of users share some common discrete subcarriers.
A major challenge associated with joint subcarrier and power allocation is to find, for a given channel state, the globally optimal {subcarrier and power} allocation strategy to minimize the total transmission power or maximize the system utility function. This paper mainly studies the computational challenges of the joint subcarrier and power allocation problem for the multi-user OFDMA system. We have shown that the general joint subcarrier and power allocation problem for the multi-user OFDMA system is strongly NP-hard. The complexity result suggests that we should abandon {efforts} to
find globally optimal subcarrier and power allocation strategy for the general multi-user OFDMA system unless for some special cases (i.e., the case when there is only one user in the system, {or} the case when the number of users is equal to the number of subcarriers). The problem is shown to be (strongly) polynomial time solvable in {these} special cases. In a companion paper, we shall design efficient algorithms to solve the joint subcarrier and power allocation problem.

\section*{Acknowledgment}
The first author wishes to thank Professor Zhi-Quan (Tom) Luo of University of Minnesota for inviting him to visit Xidian University, where Professor Luo organized a summer seminar on optimization and its applications to signal processing and communications from May 21 to August 19, 2012. This work was started from then. The first author also would like to thank Dr. Peng Liu at Xidian University {and Dr. Qiang Li at Chinese University of Hong Kong} for many useful discussions.

\appendices

\section{Proof of Theorem \ref{thm-feasi2}}\label{app-thm-feasi2}
{Before going into very details, let us first give a high level preview of the proof.
The basic idea of proving the strong NP-hardness of problem \eqref{problem-simp} with $c>1$ is to reduce it to the one with $c=2.$
 More specifically, we first partition all users into two types (Type-I and Type-II) and also all subcarriers into two types (Type-I and Type-II), where the number of Type-I subcarriers is required to be twice as large as the number of Type-I users. Then, we construct ``good'' channel parameters between Type-I (Type-II) users and Type-I (Type-II) subcarriers while ``bad'' channel parameters between Type-I (Type-II) users and Type-II (Type-I) subcarriers such that the only way for all users to satisfy their transmission rate requirements is that Type-I (Type-II) users will only transmit power on Type-I (Type-II) subcarriers. In addition, in our construction, all Type-II users and Type-II subcarriers are dummy ones, since all Type-II users' transmission rate requirements can easily be satisfied by transmitting full power on Type-II subcarriers. In this way, the problem of checking the feasibility of problem \eqref{problem-simp} with $c>1$ reduces to the one of checking whether all Type-I users' transmission rate targets can be met or not, where the number of Type-I subcarriers is twice larger than the number of Type-I users ($c=2$) as required in the partition.

}

{Below is the detailed proof of Theorem \ref{thm-feasi2}.} By Lemma \ref{thm-feasi}, we only need to consider the {following two} cases: (i) $1<c<2$ and (ii) $c>2$.

In case (i), we partition $K$ users into $(c-1)K$ Type-I users and $(2-c)K$ Type-II users, and $N=cK$ subcarriers into $2(c-1)K$ Type-I subcarriers and $(2-c)K$ Type-II subcarriers. We construct the channel parameters of Type-I users on Type-I subcarriers in the same way as in the proof of Lemma \ref{thm-feasi} where $c=2$. Furthermore,
noise powers, direct-link channel gains, and power budgets of all Type-II users on Type-II subcarriers are set to $0.3,$ $1,$ and $3,$ respectively; these parameters of all Type-I users on Type-II subcarriers and all Type-II users on Type-I subcarriers are set to $3,0.25,$ and $1,$ respectively. All Type-II users' desired transmission rate targets are set to be $\log_2 11.$  See Fig. \ref{plot}
 for the corresponding OFDMA system.

Our construction is such that the channel condition of the Type-I (Type-II) users on the Type-I (Type-II) subcarriers is reasonably better than the one of the Type-I (Type-II) users on the Type-II (Type-I) subcarriers. {One} can check that the only possible way for all $K$ users to meet their transmission rate targets is that, each Type-II user transmits full power on ({any}) one Type-II subcarrier (actually Type-II users and subcarriers are dummy users and subcarriers) and all Type-I users appropriately transmit power on Type-I subcarriers. By Lemma \ref{thm-feasi}, however, checking whether all Type-I users' transmission rate requirements can be satisfied is strongly NP-hard.

\begin{figure}[t]
     \centering
     \centerline{\includegraphics[width=9.7cm]{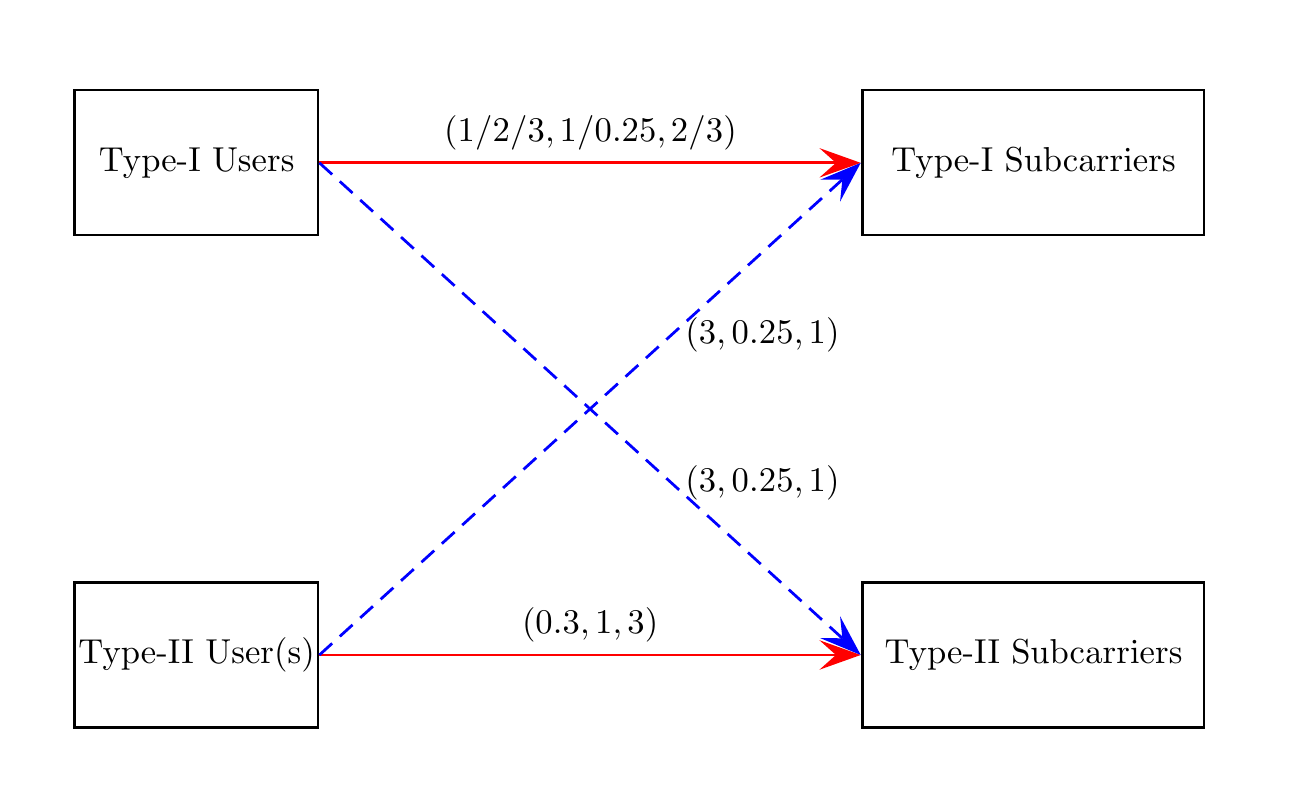}}
     \caption{An illustration of the multi-user OFDMA system constructed for case (i) and case (ii), where the number of Type-I subcarriers is exactly twice of the number of Type-I users. In both cases, the channel parameters of Type-I users on Type-I subcarriers are chosen in the same way as in the proof of Lemma \ref{thm-feasi} where $c=2$. Noise powers, direct-link channel gains, power budgets of Type-II user(s) on Type-II subcarriers are set to $0.3,$ $1,$ and $3,$ respectively; these parameters of Type-I users on Type-II subcarriers and  Type-II user(s) on Type-I subcarriers are set to $3,0.25,$ and $1,$ respectively. For case (i), the transmission rate targets of Type-II users are all set to $\log_2 11.$  For case (ii), the transmission rate target of the single Type-II user is set to  $((c-2)K+2)\log_2 11.$}
     \label{plot}
     \end{figure}



In case (ii), we partition $K$ users into $K-1$ Type-I users and $1$ Type-II user, and $N=cK$ subcarriers into $2(K-1)$ Type-I subcarriers and $(c-2)K+2$ Type-II subcarriers. We construct the channel parameters of Type-I users on Type-I subcarriers in the same way as in the proof of Lemma \ref{thm-feasi} where $c=2.$ Moreover, noise powers, direct-link channel gains, and power budgets of the single Type-II user on Type-II subcarriers are set to $0.3,$ $1,$ and $3,$ respectively; these parameters of all Type-I users on Type-II subcarriers and the single Type-II user on Type-I subcarriers are set to $3,0.25,$ and $1,$ respectively. The transmission rate of the single Type-II user is required to be not less than $((c-2)K+2)\log_2 11.$  See Fig. \ref{plot} for the corresponding system.

Due to special construction of the system, {one} can check that the only way for all $K$ users to meet their transmission rate targets is that, the single Type-II user transmits full power on all Type-II subcarriers and all Type-I users appropriately transmit power on Type-I subcarriers. Again, by Lemma \ref{thm-feasi}, checking whether all Type-I users' transmission rate requirements can be satisfied is strongly NP-hard.

\section{Proof of Theorem \ref{thm-uti2}}\label{app-thm-uti2}
We first prove the strong NP-hardness of problem \eqref{utility} {with $H=H_1,H_2,H_3,~\text{and}~H_4$} for the special case $c=2$ and then prove their strong NP-hardness for the general case $c>1.$

{\textbf{We first consider the case $N/K=2$.}} For any instance of the $3$-dimensional matching problem, we construct the same system as in the proof of Lemma \ref{thm-feasi} and set $P_{k}=5~(k\in\K=\left\{1_x,2_x,...,K_x\right\}).$

{\textbf{Strong NP-hardness of problem \eqref{utility} with $H=H_4:$}} Lemma \ref{thm-feasi} directly implies that the following problem
   \begin{equation*}\label{min}
 \begin{array}{cl}
\displaystyle \max_{\{p_k^n\}} & \displaystyle \min_{k\in\K} \left\{R_{k}\right\} \\[8pt]
\text{s.t.} 
    & \displaystyle \sum_{n\in\N}p_k^{n}\leq P_{k},~k\in\K, \\[12pt]
    & P_k^n\geq p_k^n\geq 0,~k\in\K,~n\in\N,\\[5pt]
    & p_k^np_j^n=0,~\forall~j\neq k,{~k,\,j\in\K,}~n\in\N
    \end{array}
\end{equation*}
is strongly NP-hard, since the problem of checking whether its optimal value is greater than or equal to $3$ is strongly NP-hard. 

{\textbf{Strong NP-hardness of problem \eqref{utility} with $H=H_1:$}} {We prove that the sum-rate maximization problem
   \begin{equation}\label{sum}
 \begin{array}{rl}
\displaystyle \max_{\{p_k^n\}} & \displaystyle \frac{1}{K}\sum_{k\in\K} R_{k_x} \\[12pt]
\text{s.t.} 
    & \displaystyle \sum_{n\in\N}p_k^{n}\leq P_{k},~k\in\K, \\[12pt]
    & P_k^n\geq p_k^n\geq 0,~k\in\K,~n\in\N,\\[5pt]
    & p_k^np_j^n=0,~\forall~j\neq k,{~k,\,j\in\K,}~n\in\N
    \end{array}
\end{equation} is also strongly NP-hard. In particular, we show that checking the optimal value $H_1(R_{1_x}^*,R_{2_x}^*,...,R_{K_x}^*)$ of problem \eqref{sum} is greater than or equal to $3$ is strongly NP-hard, where $R_{k_x}^*$ is the transmission rate of user $k_x$ at the solution of problem \eqref{sum}.} To this aim, consider the following relaxation of problem \eqref{sum},
   \begin{equation}\label{relaxation}
 \begin{array}{cl}
\displaystyle \max_{\{p_k^n\}} & \displaystyle \frac{1}{K}\sum_{k\in\K} R_k' \\[12pt]
\text{s.t.} 
    & \displaystyle \sum_{k\in\K}\sum_{n\in\N}p_k^{n}\leq 5K, \\[13pt]
    & p_k^n\geq 0,~k\in\K,{~k,\,j\in\K,}~n\in\N,\\[5pt]
    & p_k^np_j^n=0,~\forall~j\neq k,~n\in\N,
    \end{array}
\end{equation} where $$R_k'=\sum_{n\in\Y}\log_2\left(1+\frac{p_k^n}{1}\right)+\sum_{n\in\Z}\log_2\left(1+\frac{p_k^n}{2}\right)\geq R_k.$$
Due to the OFDMA constraint, we see that problem \eqref{relaxation} is equivalent to

\begin{equation}\label{relaxation2}
 \begin{array}{cl}
\displaystyle \max_{\left\{p^{n}\right\}} & \displaystyle \sum_{n\in\Y}\log_2\left(1+\frac{p^n}{1}\right)+\sum_{n\in\Z}\log_2\left(1+\frac{p^n}{2}\right) \\[13pt]
\text{s.t.} & \displaystyle \sum_{n\in\N}p^{n}\leq 5K,\\[12pt]
& p^n\geq 0,~n\in\N.
\end{array}
\end{equation} Noticing that problem \eqref{relaxation2} is convex, we can obtain its optimal solution $$(p^{1_y},p^{2_y},...,p^{K_y},p^{1_z},p^{2_z},...,p^{K_z})=(3,3,...,3,2,2,...,2)$$ and its optimal value $3K.$
Therefore, the optimal value $H_1(R_{1_x}^*,R_{2_x}^*,...,R_{K_x}^*)$ of the original problem \eqref{sum} is less than or equal to $3,$ and the equality holds if and only if $$R_{1_x}^*=R_{2_x}^*=\cdots=R_{K_x}^*=3.$$
The latter holds if and only if the answer to the $3$-dimensional matching problem is yes. Therefore, the optimal value of problem \eqref{sum} is greater than or equal to $3$ if and only if the answer to the $3$-dimensional matching problem is yes.

{\textbf{Strong NP-hardness of problem \eqref{utility} with $H=H_2$ and $H=H_3:$}} For the cases $H_2$ and $H_3$, notice that for all $R_1$, $R_2$, \ldots, $R_K\geq 0,$
\begin{align*}&H_4(R_1,R_2,...,R_K)\leq H_3(R_1,R_2,...,R_K)\\
\leq ~& H_2(R_1,R_2,...,R_K)\leq H_1(R_1,R_2,...,R_K)\end{align*}
 and the equalities hold if and only if $$R_1=R_2=\cdots=R_K.$$ Therefore, the optimal value of $H_2(R_{1_x}^*,R_{2_x}^*,...,R_{K_x}^*)$ or $H_3(R_{1_x}^*,R_{2_x}^*,...,R_{K_x}^*)$ is greater than or equal to $3$ if and only if the answer to the $3$-dimensional matching problem is yes. This implies the strong NP-hardness of all the four utility maximization problems {\eqref{utility} with $c=2$}.

{Now we consider the general case $N/K=c>1.$} We show the strong NP-hardness of problem \eqref{utility} for the general case $c>1$ by constructing some dummy users and subcarriers as in the proof of Theorem \ref{thm-feasi2}.
~Take the sum-rate maximization problem as an example. It is simple to check that
\begin{itemize}
  \item [-] in case $1<c<2,$ the sum-rate utility function of the constructed system is greater than or equal to $$3(c-1)+(2-c)\log_2 11$$ if and only if the given instance of the $3$-dimensional matching problem with size $(c-1)K$ has a positive answer;
      \item [-] in case $c>1,$ the sum-rate utility function of the constructed system is greater than or equal to $$\frac{3(K-1)+((c-2)K+2)\log_2 11}{K}$$ if and only if the given instance of the $3$-dimensional matching problem with size $K-1$ has a positive answer.
\end{itemize}Hence, the sum-rate maximization problem in the case that $c>1$ is strongly NP-hard. Similar results also hold true
for the other three utility functions. We omit the proof for brevity.

\section{Extended Water-Filling Solution for Power Control Problem \eqref{sub2}}\label{app-2}
%

{To show \eqref{pkn2} is the solution to problem \eqref{sub2}, let us first write down the KKT condition of problem \eqref{sub2}. Suppose $\left(\lambda,~\left\{\xi^n\right\},~\left\{\nu^n\right\}\right)$ are the Lagrangian multipliers corresponding to the constraints $\sum_{n\in\N}p^n\leq P,$ $P^n\geq p^n,$ and $p^n\geq 0,$ respectively, then the KKT condition of problem \eqref{sub2} is given as follows
  \begin{equation}\label{KKT}\left\{\!\!\!\!\!
    \begin{array}{cl}
      & \frac{-1}{\eta^n/\alpha^n+p^n}+\lambda+\xi^n-\nu^n=0,\\[10pt]
      & \displaystyle \sum_{n\in\N}p^n\leq P,~\lambda\geq 0,~\lambda\left(\sum_{n\in\N}p^n-P\right)=0,\\[15pt]
      &  P^n\geq p^n,~\xi^n\geq 0,~(P^n-p^n)\xi^n=0,~n\in\N,\\[6pt]
      &  p^n\geq 0,~\nu^n\geq0,~p^n\nu^n=0,~n\in\N.
    \end{array}\right.
  \end{equation}
  Since problem \eqref{sub2} is convex for any fixed $P>0$, it follows that the KKT condition \eqref{KKT} is necessary and sufficient for $\left\{p^n\right\}$ to solve problem \eqref{sub2}. Therefore, to show \eqref{pkn2} is the solution to problem \eqref{sub2}, it suffices to show that there exist appropriate Lagrangian multipliers such that $\left(\left\{p^n\right\}, \lambda, \left\{\xi^n\right\}, \left\{\nu^n\right\}\right)$ satisfy the KKT system \eqref{KKT}. Next, we construct appropriate $(\lambda, \left\{\xi^n\right\}, \left\{\nu^n\right\})$ such that $\left\{p^n\right\}$ in \eqref{pkn2} together with the constructed $(\lambda, \left\{\xi^n\right\}, \left\{\nu^n\right\})$ satisfy all the conditions in \eqref{KKT}.}

Specifically, we choose $\lambda\geq 0$ such that the objective value of problem \eqref{sub2} is equal to $\gamma,$ and
  \begin{itemize}
    \item [-] if $p^n=0$ in \eqref{pkn2}, then $1/\lambda-\eta^n/\alpha^n\leq 0;$ we set $\xi^n=0$ and $\nu^n=\lambda-\alpha^n/\eta^n\geq 0;$
    \item [-] if $0<p^n<P^n$ in \eqref{pkn2}, then $P^n>1/\lambda-\eta^n/\alpha^n> 0;$ we set $\xi^n=\nu^n=0;$
    \item [-] if $p^n=P^n$ in \eqref{pkn2}, then $1/\lambda-\eta^n/\alpha^n\geq P^n;$ set $\xi^n=\frac{1}{\eta^n/\alpha^n+P^n}-\lambda\geq 0$ and $\nu^n=0.$
  \end{itemize}It is simple to check $\left(\left\{p^n\right\},\lambda,\left\{\xi^n\right\},\left\{\nu^n\right\}\right)$ constructed in the above satisfies the KKT system \eqref{KKT}.
~Hence, \eqref{pkn2} is the solution to problem \eqref{sub2}.

\section{The Order of Stirling Number $S(K+C,K)$}\label{app-s}
In combinatorics, a Stirling number of the second kind \cite{stirling}, denoted by $S(N,K),$ is the number of ways to partition a set of $N$ objects into $K$ nonempty subsets, where $N\geq K\geq1$. Stirling numbers of the second kind obey the following recursive relation
\begin{equation}\label{stirling}S(N+1,K)=S(N,K-1)+KS(N,K)\end{equation} with initial conditions $S(k,k)=1$ and $S(k,0)=0$ for any $k\geq 1.$ To understand this formula, observe that a partition of $N+1$ objects into $K$ nonempty subsets either contains the $(N+1)$-th object as a singleton (which corresponds to the term $S(N,K-1)$ in \eqref{stirling}) or contains it with some other elements (which corresponds to the term $KS(N,K)$ in \eqref{stirling}).

Next, we claim by induction that $$S(K+C,K)\leq(K+C)^{2C}.$$ In fact, when $C=1,$ we obtain $$S(K+1,K)=\frac{K(K+1)}{2}.$$ This is because that dividing $K+1$ elements into $K$ sets means dividing it into one set of size $2$ and $K-1$ sets of size $1.$ Therefore, we only need to pick those two elements. Assume that $S(K+C-1,K)\leq {\left(K+C-1\right)}^{2(C-1)},$ we show that $S(K+C,K)\leq (K+C)^{2C}.$ By invoking \eqref{stirling}, we have\begin{equation*}
\begin{array}{rl}
&S(K+C,K)\\[3pt]
=&S(K+C-1,K-1)+K{S(K+C-1,K)}\\[3pt]
\leq&S(K+C-1,K-1)+K{\left(K+C-1\right)}^{2(C-1)}\\[3pt]
\leq&S(K+C-1,K-1)+{\left(K+C\right)}^{2C-1}\\[3pt]
\leq&\sum_{k=1}^K \left(k+C\right)^{2C-1}\\[3pt]
\leq&(K+C)^{2C},
\end{array}\end{equation*}which shows that $S(K+C,K)\leq (K+C)^{2C}$ holds true.

\end{document}